\documentclass[final,3p,times,twocolumn,authoryear]{elsarticle}

\usepackage{graphicx} 

\usepackage{tikz}
\usetikzlibrary{shapes.geometric}

\usepackage{amsfonts,amsmath,amssymb,bm}

\usepackage{amsthm}

\newtheorem{problem}{Problem}
\newtheorem{theorem}{Theorem}
\newtheorem{lemma}{Lemma}
\newtheorem{remark}{Remark}
\newtheorem{definition}{Definition}
\newtheorem{assumption}{Assumption}

\DeclareMathOperator{\clamp}{clamp}

\usepackage{blindtext}

\usepackage{booktabs}

\journal{Control Engineering Practice}

\begin{document}

\begin{frontmatter}


\title{A Long-Duration Autonomy Approach to Connected and Automated Vehicles} 


\author{Logan E. Beaver} 

\affiliation{organization={Old Dominion University},
            addressline={Department of Mechanical and Aerospace Engineering}, 
            city={Norfolk},
            postcode={VA}, 
            state={23529},
            country={USA}}

\begin{abstract}
In this article, we present a long-duration autonomy approach for the control of connected and automated vehicles (CAVs) operating in a transportation network.
In particular, we focus on the performance of CAVs at traffic bottlenecks, including roundabouts, merging roadways, and intersections.
We take a principled approach based on optimal control, and derive a reactive controller with guarantees on safety, performance, and energy efficiency.
We guarantee safety through high order control barrier functions (HOCBFs), which we ``lift'' to first order CBFs using time-optimal motion primitives.
This yields a set of first-order CBFs that are compatible with the control bounds.
We demonstrate the performance of our approach in simulation and compare it to an optimal control-based approach.
\end{abstract}


\end{frontmatter}



\section{Introduction}

Connected and automated vehicles (CAVs) continue to proliferate transportation networks. As a result, it is critical for us to develop control algorithms that are computationally efficient, provably safe, and produce energy-efficient trajectories.
This idea is consistent with long-duration autonomy, where robotic agents are left to operate in the field for months to years without regular human feedback.
Thus, the purpose of this article is to bridge the gap between optimal CAV motion planning and long-duration autonomy techniques.

One emerging approach to long-duration autonomy is constraint-driven (or ecologically-inspired) control for cyber-physical systems \citep{Egerstedt2018RobotAutonomy}.
For long-duration autonomy tasks, robots are left to interact with their environment on timescales significantly longer than what can be achieved in a laboratory setting.
These approaches necessarily emphasize safe energy-minimizing control policies for the agents, whose behaviors are driven by interactions with the environment.
Several applications of constraint-driven multi-agent control have been explored recently \citep{Notomista2019AnSystems,Ibuki2020Optimization-BasedBodies,beaver2023constraint}.
Under this constraint-driven paradigm, each CAV seeks to minimize its total energy consumption, subject to a set of safety and task constraints.

In this article, we propose a constraint-driven approach to control CAVs operating alongside other CAVs and human drivers while minimizing energy consumption and guaranteeing safety.
Long-term energy savings has been partially addressed by the platooning literature, where developments in cooperative adaptive cruise control \citep{wang2017developing,ames2014control} and mixed-traffic platooning \citep{mahbub2021platoon} have shown promise.
Decentralized algorithms, such as Reynolds' flocking, have been applied for highway vehicles \citep{fredette2017fuel} to minimize energy consumption while maintaining a desired speed; a recent review of these techniques is presented in \citep{beaver2021overview}.
However, a CAV operating over a long duration will spend a significant amount of time at traffic bottlenecks \citep{chalaki2022research}, such as intersections, roundabouts, and merging zones.
Existing platooning controllers are generally not designed with the safety constraints required to avoid lateral collisions.

Additional approaches have used optimal control to generate CAV trajectories through traffic bottlenecks \citep{sabouni2024optimal,malikopoulos2021optimal}, which can include treating platoons as one large vehicle \citep{mahbub2023constrained}.
However, these solutions can be computationally expensive, and can be challenging to incorporate with human driven vehicles.
As a result, these approaches are not generally applicable far from traffic bottlenecks, where the CAVs may switch to another control mode \citep{bang2022combined}.

In contrast, this work is motivated by long-duration autonomy to derive a single control algorithm that is feasible throughout the entire transportation network.
The controller is reactive, and thus we provide strong guarantees on safety without requiring an explicit predictive model of other vehicles' trajectories using high-order control barrier functions (HOCBFs) \citep{xiao2019control}.
This means that our approach is compatible with signalized intersections, human driven vehicles, and mixed CAV-human traffic.
Our approach is also decentralized, and thus it is well-suited to ``open systems'' where other agents may suddenly enter or exit the system.
The contributions of this work are as follows:
\begin{itemize}
    \item A framework to map infinite-horizon optimal control problems to simple reactive controllers (Theorem \ref{thm:lqr}).
    \item A constraint-driven controller for long-duration CAV autonomy (Problem \ref{prb:reactive}).
    \item A motion-primitives approach to HOCBFs that guarantees both safety and actuation constraints (Theorem \ref{thm:CBFs}).
\end{itemize}

The remainder of the article is organized as follows.
We briefly discuss notation next, followed by the problem formulation in Section \ref{sec:problem}.
We convert our optimal control problem into a reactive controller in Section \ref{sec:policy}, and prove its safety and performance.
We demonstrate our approach in Section \ref{sec:sim}, and compare the results with an optimal control-based solution for 10 CAVs at an unsignalized intersection.
Finally, we draw conclusions and propose some directions for future research in Section \ref{sec:conclusion}.

\subsection{Notation}

Many classic references on optimal control, e.g., \citep{Bryson1975AppliedControl,Ross2015}, consider centralized optimal control problems.
Thus, directly adopting their notation may lead to ambiguities about the state space of a decentralized problem.
To relieve this tension, we take the following approach for an agent with index $i$.
Endogenous variables, e.g., the position of agent $i$, are written without explicit dependence on time.
Exogenous variables, e.g., the position of agent $j$ as measured by agent $i$, are written with an explicit dependence on time.
This notation is common in the applied mathematics literature \citep{Levine2011OnFlatness}, and makes it explicitly clear how functions evolve with respect to the state (e.g., state dynamics) and how they evolve with respect to time (e.g., external signals measured by the agent).

\section{Problem Formulation} \label{sec:problem}

Consider a collection of $N$ CAVs traveling through a transportation network.
We index each CAV with a subscript $i\in\{1,2,\dots,N\}$; each CAV has two states $p_i, v_i\in\mathbb{R}$ and a control action $u_i\in\mathbb{R}$.
These correspond to the longitudinal position, velocity, and acceleration of the CAV along some fixed path.
To model the CAV dynamics, we employ a double-integrator model,
\begin{equation} \label{eq:dynamics}
    \begin{bmatrix}
        \dot{p}_i \\ \dot{v}_i
    \end{bmatrix}
    = 
    \begin{bmatrix}
        0 & 1 \\ 0 & 0
    \end{bmatrix}
    \begin{bmatrix}
        p_i \\ v_i
    \end{bmatrix}
    +
    \begin{bmatrix}
        0 \\ 1
    \end{bmatrix}
    u_i.
\end{equation}

Our objective is the long duration deployment of CAVs in urban and highway environments.
To this end, each CAV seeks to optimize an infinite-horizon cost,
\begin{equation} \label{eq:Jcost}
    J(p_i, v_i, u_i) = \frac{1}{2}\int_{0}^{\infty} (v_i - v_i^d)^2 + \frac{1}{\alpha^2} u_i^2\,dt,
\end{equation}
where $v_i^d$ is a desired driving speed, e.g., the energy-optimal speed or the posted speed limit, $\alpha$ is a regularizing parameter to avoid impulse-like accelerations, and we use the fractions $1/2$ and $1/\alpha$ to simplify the notation of our optimal solution.
Each CAV is also subject to the control constraints,
\begin{align} \label{eq:bounds}
    |u_i|& \leq u_{\max},
\end{align}
which correspond to the maximum safe acceleration rates.
We impose a rear-end safety constraint between sequential vehicles,
\begin{equation} \label{eq:rear-end}
    \delta_i(t) - p_i \geq \gamma_i,
\end{equation}
where $\delta_i(t)$ measures the position of the preceding vehicle (if it exists), and $\gamma_i$ is a fixed standstill distance.
Note that, in the absence of a preceding vehicle, we let $\delta_i\to\infty$, and \eqref{eq:rear-end} is satisfied.


We formulate our problem based on the following working assumptions about the transportation network and vehicles.
\begin{assumption}\label{smp:nominal}
    Far from major traffic bottlenecks (e.g., merging zones, roundabouts, and intersections), the flow of CAVs is unimpeded and efficient.
\end{assumption}
Assumption \ref{smp:nominal} is common in the literature, as it enables designers to focus on single \citep{malikopoulos2021optimal} or multiple connected \citep{chalaki2022research} bottlenecks.
This assumption is justified as long as traffic far from the bottlenecks remains below the critical road capacity.

\begin{assumption} \label{smp:tracking}
    Each CAV is equipped with a tracking controller that handles noise, disturbances, model mismatch, and other errors.
\end{assumption}

\begin{assumption} \label{smp:comms}
    Communication between the CAVs and smart infrastructure is instantaneous with no noise, delay, or packet loss.
\end{assumption}

Assumptions \ref{smp:tracking} and \ref{smp:comms} are not restrictive on our approach, we only employ them to simplify our analysis.
Both assumptions can be relaxed for a real system using learning \citep{Greeff2021ExploitingProcesses} and robust control \citep{Emam2019RobustTestbed} methods.

Under Assumption \ref{smp:nominal}, we consider the $N$ CAVs traveling through an intersection, which is the most complex traffic bottleneck.
As each CAV is traveling on a fixed path, we can geometrically determine points along each path where collisions between CAVs may occur.
We refer to these points as \textit{collision nodes}, and we consider them at every point where two paths cross.
For each CAV $i$ in the intersection, we compute the number of collision nodes $K_i$.
This is the basis for our notion of a \textit{scheduled intersection}.

\begin{definition}[scheduled intersection] \label{def:psi}
A \textit{scheduled intersection} is an intersection where each approaching CAV $i$ is given explicit non-zero time interval(s) $\big[\underline{t}_i^k, \overline{t}_i^k\big] \in \mathbb{R}$ to cross conflict nodes $k = 1, 2, \dots, K_i$. 
\end{definition}

As we discuss in the sequel, computing the interval $\big[\underline{t}_i^k, \overline{t}_i^k\big]$ for CAV $i$ to cross conflict point $k\in\{1,2,\dots,K_i\}$ depends on the implementation of the scheduled intersection.
This may be the interval where a smart traffic signal will be  green, or it may be determined by an intelligent roadside unit, i.e., a \textit{coordinator} acting as a database.
With this definition in place, we formalize the control algorithm of the CAVs with Problem \ref{prb:original}.

\begin{problem} \label{prb:original}
    Generate the trajectory that optimizes,
    \begin{align*}
        \min_{T,u_i(t)}  &\int_0^T \frac{1}{2} (v_i-v_i^d)^2 + \frac{1}{\alpha^2} u_i^2 dt \\
    \text{subject to: } & \\
    \text{dynamics \& safety:}& \quad\eqref{eq:dynamics},\,\eqref{eq:bounds},\,\eqref{eq:rear-end}, \\
    \text{earliest arrival:}& \quad p(t) + \gamma_i \leq p(\underline{t}_i^k) \quad \forall t \leq \underline{t}_i^k,    \\
    \text{latest departure:}&\quad p(t) \geq p(\overline{t}_i^k) + \gamma_i \quad \forall t \geq \overline{t}_i^k, \\
    \text{given:}&\quad
    \bm{x}(t_0),\, \underline{t}_i^k, \overline{t}_i^k,
    \end{align*}
    for each collision node $k = 1, 2, \dots, K_i$ and over a fixed planning horizon $T$.
\end{problem}

Solving Problem \ref{prb:original} yields an optimal and safe trajectory, however, it requires having an explicit model of other vehicles' (possibly unknown) trajectories for $t\in[0,T]$ to estimate $\delta_i(t)$.
Next, we describe our optimization-based control policy to generate CAV trajectories based on Problem \ref{prb:original}.

\section{Optimization-Based Control Policy} \label{sec:policy}

We develop our control policy based on Problem \ref{prb:original} in two stages: (i) we derive the unconstrained optimal trajectory and find an equivalent feedback law, (ii) we guarantee constraint satisfaction using control barrier functions (CBFs) while minimizing deviation from the optimal feedback law.
To generate the optimal feedback law, we first consider the unconstrained solution to Problem \ref{prb:original}.
Because the cost \eqref{eq:Jcost} is convex, Pontryagin's minimum principle yields the \textbf{globally optimal} solution when no constraints are active.
Applying Pontryagin's minimum principle yields the optimality condition \citep{beaver2024optimal,beaver2023lq},
\begin{equation} \label{eq:ode}
    \ddot{u} - \alpha^2 u = 0.
\end{equation}

The optimality condition \eqref{eq:ode} is a linear ordinary differential equation with constant coefficients.
It yields a straightforward analytical solution,
\begin{equation} \label{eq:mp}
\begin{aligned}
    p_i(t) &= c_1 e^{\alpha t} + c_2 e^{-\alpha t} + c_3 t + c_4, \\
    v_i(t) &= c_1 \alpha e^{\alpha t} - c_2 \alpha e^{-\alpha t} + c_3, \\
    u_i(t) &= c_1 \alpha^2 e^{\alpha t} + c_2 \alpha^2 e^{-\alpha t},
\end{aligned}
\end{equation}
where $c_1$--$c_4$ are unknown constants of integration.
In existing approaches, the boundary conditions of Problem \ref{prb:original} are used to find the particular constants of integration in \eqref{eq:mp}.
This yields the optimal unconstrained trajectory for the CAV to track.
Instead, we use the following result to derive an equivalent feedback law to control the CAV.

\begin{theorem} \label{thm:lqr}
    As the planning horizon $T\to\infty$, the optimal unconstrained trajectory \eqref{eq:mp} approaches the feedback law,
    \begin{equation} \label{eq:lqr}
        u_i(\bm{x}_i) = \alpha\Big(v_i^d - v_i\Big),
    \end{equation}
    where $v_i^d$ is a desired steady-state velocity for CAV $i$.
\end{theorem}

\begin{proof}
We prove Theorem \ref{thm:lqr} by assuming that the cost \eqref{eq:Jcost} is bounded, which implies that $v_i$ and $u_i$ are both finite.
First, $T\to\infty$ implies that $c_1 \to 0$.
Next, the steady-state velocity of the CAV must satisfy,
\begin{equation}
     c_3 = \lim_{t\to\infty} v_i(t) = v_i^d.
\end{equation}
Dividing \eqref{eq:mp} by $\alpha$ and re-arranging yields,
\begin{equation} \label{eq:feedback}
    \frac{u(t)}{\alpha} = c_2\alpha e^{-\alpha t} = c_3 -v_i(t) =  v_i^d - v_i(t).
\end{equation}    
Multiplying both sides by $\alpha$ completes the proof.
\end{proof}

We note that while the proof of Theorem \ref{thm:lqr} is equivalent to solving the infinite-horizon LQR problem, however, we do so without first generating and solving the steady-state Riccati equation.
Furthermore, we interpret the feedback law \eqref{eq:feedback} as continuously re-planning the CAV's trajectory at every time instant, i.e., for any state $\bm{x}_i$, \eqref{eq:feedback} induces the optimal unconstrained trajectory.
Next, we construct a reactive controller that uses \eqref{eq:lqr} as a reference input, and we derive CBFs to enforce the constraints imposed on Problem \ref{prb:original}.

\begin{lemma} \label{lma:primitives}
    The conditions
    \begin{align}
        v_i \geq& \frac{\Delta p}{\Delta t} - \frac{1}{2} u_{\max} \Delta t, \label{eq:kin-1} \\
        v_i \leq& \frac{\Delta p}{\Delta t} + \frac{1}{2} u_{\max} \Delta t, \label{eq:kin-2}
    \end{align}
    guarantee that the CAV travels a distance $\Delta p$ before (after) the time interval $\Delta t$ elapses, i.e., they are the latest departure and earliest arrival deadlines, respectively.
\end{lemma}

\begin{proof}
    To prove Lemma \ref{lma:primitives}, we start with the kinematic equation for constant acceleration over an interval $[t, t + \Delta t]$,
    \begin{equation}
        p(t+\Delta t) = p_i(t) + v_i(t) \Delta t + \frac{1}{2} u_{\max} (\Delta t)^2.
    \end{equation}
    Substituting  $\Delta p := p_i(t+\Delta t) - p_i(t)$ and rearranging yields,
    \begin{equation}
        v_i(t) = \frac{\Delta p}{\Delta t} - \frac{1}{2} u_{\max} \Delta t,
    \end{equation}
    which satisfies the equality of \eqref{eq:kin-1}.
    Substituting $v_i'(t) > v_i(t)$ yields a new change in position,
    \begin{equation}
        \Delta p_i' = v_i'(t) \Delta t + \frac{1}{2} u_{\max} \Delta t^2 > \Delta p_i,
    \end{equation}
    which implies the inequality in \eqref{eq:kin-1}.
    The proof of \eqref{eq:kin-2} is identical, and we omit it for brevity.
\end{proof}

\begin{remark}[Lemma 2 in \citep{tzortzoglou2025mixed}] \label{rmk:overshoot}
We impose the condition
\begin{equation} \label{eq:overshoot}
    \Delta t^* \leq \sqrt{\frac{2\Delta p}{u_{\max}}}
\end{equation}
to ensure that the CAV does not overshoot $\Delta p$ when applying condition \eqref{eq:kin-2}.
\end{remark}

Remark \ref{rmk:overshoot} is significant, because \eqref{eq:kin-2} only implies that the CAV will travel a distance $\Delta p$ at time $\Delta t$ or later.
This can be achieved two ways: (i) decelerating with $v(t) \geq 0$ and reaching $\Delta p$ at time $\Delta t$; (ii) applying a smaller acceleration, overshooting $\Delta p$, stopping, then reversing to reach point $\Delta p$ at time $\Delta t$.
The second case is undesirable, and Remark \ref{rmk:overshoot} is derived by enforcing $v(\Delta t^*)\geq 0$.
Thus, when $\Delta t > \Delta t^*$, we require the vehicle to stop by $\Delta t^*$ then remain stopped for $t\in[\Delta t^*, \Delta t]$.
Without loss of generality, we omit $t^*$ in the remainder of the article to simplify our analytical results.

\begin{theorem} \label{thm:CBFs}
    Let CAV $i$ approach a collision node with index $k$ at position $p_i^k$ and the crossing time interval $[\underline{t}_i^k, \overline{t}_i^k]$.
    Then, the CBFs,
    \begin{align}
        u \leq\,& -\kappa_T\Big(v_i - \frac{\Delta p}{\Delta t_1} - u_{\max}\frac{\Delta t_1}{2}\Big) \notag\\
        &+ \frac{\Delta p - v_i \Delta t_1}{\Delta t_1^2} - \frac{u_{\max}}{2}, \label{eq:cbf-t1} \\
         u_i \geq\,& \kappa_T\Big(\frac{\Delta p}{\Delta t_2} - u_{\max} \frac{\Delta t_2}{2} - v_i\Big) \notag\\
       &+ \frac{\Delta p - v_i \Delta t_2}{\Delta t_2^2} + \frac{u_{\max}}{2}, \label{eq:cbf-t2}
    \end{align}
    where $\Delta p := p_i^k - p(t)$ is distance to the collision node, $\Delta t_1 := \underline{t}_i^k - t$, $\Delta t_2 := \overline{t}_i^k - t$, are the earliest arrival and latest departure times, respectively, and $\kappa_T > 0$ is a gain.
    These CBFs guarantee the arrival of CAV $i$ at position position $p_i^k$ during the interval $[\underline{t}_i^k, \overline{t}_i^k]$.
\end{theorem}

\begin{proof}
    We prove Theorem \ref{thm:CBFs} by construction.
    We enforce \eqref{eq:kin-1} and \eqref{eq:kin-2} using a linear CBF of the form,
    \begin{equation} \label{eq:barrier}
        \dot{b}(x) \leq -\kappa b(x),
    \end{equation}
    which guarantees forward invariance of $b(x) \leq 0$ by the Nagumo theorem \citep{Ames2019ControlApplications}.
    Rewriting \eqref{eq:kin-1} implies the constraint,
    \begin{equation} \label{eq:kin1-notation}
        b(x) = \frac{\Delta p}{\Delta t_2} - u_{\max}\frac{\Delta t_2}{2} - v_i \leq 0.
    \end{equation}
    Substituting \eqref{eq:kin1-notation} into \eqref{eq:barrier} yields,
    \begin{align}
    -u_i &+ \frac{\Delta p - v_i \Delta t_2}{\Delta t_2^2}
    + \frac{u_{\max}}{2} \\
        &\leq -\kappa_T\Big(\frac{\Delta p}{\Delta t_2} - u_{\max}\frac{\Delta t_2}{2} - v_i\Big),
    \end{align}
    and re-arranging yields \eqref{eq:cbf-t2}.
    The proof for \eqref{eq:cbf-t1} is identical, and we omit it for brevity.
\end{proof}

Finally, to guarantee the rear-end safety constraint \eqref{eq:rear-end}, we enforce the minimum stopping distance constraint \citep{beaver2021constraint} as a control barrier function,
\begin{align} \label{eq:cbf-rear-end}
    u \leq& -\kappa_R\Big( v - \dot{\delta}(t) - \sqrt{-2 u_{\max}(p - \delta(t) + \gamma)}  \Big) \notag\\
    &- \frac{u_{\max}(v - \dot{\delta})}{\sqrt{-2 u_{\max}(p - \delta(t) + \gamma)}},
\end{align}
where $\delta_i(t)$ is position of the CAV preceding $i$, if it exists, and $\kappa_R$ is the CBF gain.
Note that $\delta_i(t)$ and its derivatives may be discontinuous, i.e., if another CAV turns into the lane in front of CAV $i$.
In the even that no CAV preceds $i$, we let $\delta_i(t)\to\infty$, and \eqref{eq:cbf-rear-end} is trivially satisfied.

\begin{remark} \label{rmk:hocbf}
    The result of Theorem \ref{thm:CBFs} and \eqref{eq:cbf-rear-end} effectively ``lift'' the safety constraints in Problem \ref{prb:original} from functions of position to functions of speed while respecting the acceleration limits.
    This enables us to enforce the safety constraints throgh a first order CBF instead of explicitly considering higher order control barrier functions.
\end{remark}

With the optimal unconstrained trajectory and safety constraints defined, we present a reactive form of Problem \ref{prb:original}, where the CAV continuously replans using the infinite-horizon policy with safety enforced through our crossing time and rear-end safety CBFs.

\begin{problem} \label{prb:reactive}
At each time, CAV $i$ applies the control action that satisfies,
    \begin{align*}
    \min_{u_i}\,&  (u_i - \alpha(v_i^d - v_i))^2 \\
    \text{subject to:}\\
    \text{crossing constraints:}\, &\eqref{eq:overshoot}, \eqref{eq:cbf-t1}, \eqref{eq:cbf-t2} \\
    \text{rear-end safety:}\, &\eqref{eq:cbf-rear-end}, \\
     \text{control bounds:}\,&|u| \leq u_{\max}, \\
     \text{given:}\,&(p^k, \underline{t}_i^k, \overline{t}_i^k),\,k=1,2,\dots,K,
\end{align*}
where $[\underline{t}_i^k, \overline{t}_i^k]$ is the scheduled time interval to pass point $p_k$ of collision node $k$.
\end{problem}

\begin{theorem} \label{thm:sln}
    The solution to Problem \ref{prb:reactive} is,
    \begin{equation}
    \begin{aligned} \label{eq:clamp}
        u_i^* &= \clamp(\alpha(v_i^d - v_i), \underline{u}, \overline{u})\\
        &= \min(\max(\alpha(v_i^d - v_i),\, \underline{u}),\, \overline{u}),
    \end{aligned}
    \end{equation}
    where
    \begin{align}
        \overline{u} &= \max\big(\min(\,\eqref{eq:cbf-t1},\, \eqref{eq:cbf-rear-end}\,), -u_{\max}\big) \label{eq:u-lb}, \\
        \underline{u} &= \min\big(\eqref{eq:cbf-t2}, u_{\max}\big).  \end{align}
\end{theorem}

\begin{proof}
    The globally optimal (unconstrained) solution to Problem \ref{prb:reactive} is $u_i^* = \alpha(v_i^d - v_i)$.
    All the constraints in Problem \ref{prb:reactive} are lower and upper bounds on $u_i$, and the cost is quadratic in $u_i$,
    Thus, $u_i^*$ must be the unconstrained optimal solution; if that is infeasible, then $u_i^*$ must be the highest lower bound or the lowest upper bound.
\end{proof}

As a result of Theorem \ref{thm:sln}, Problem \ref{prb:reactive} presents an approximation of Problem \ref{prb:original} that is efficient, safe, and has the algorithmically trivial solution presented in Theorem \ref{thm:sln}.
Furthermore, the feedback control law (Theorem \ref{thm:lqr}) effectively re-plans the CAV trajectory continuously.

One major open problem with CBF-based solutions is that Problem \ref{prb:reactive} can become infeasible for certain states, e.g., when the lower and upper control bounds of \eqref{eq:clamp} become inconsistent.
We have already addressed the inconsistency between the safety constraints and control bounds in Remark \ref{rmk:hocbf}, and we address inconsistencies between the safety and crossing constraints next.

\begin{remark} \label{rmk:infeasible}
    If Problem \ref{prb:reactive} becomes infeasible, then the rear-end safety constraint \eqref{eq:cbf-rear-end} is inconsistent with the latest departure time constraint \eqref{eq:cbf-t2}.
    Resolving this is problem-dependent, but generally involves relaxing the departure constraint \eqref{eq:cbf-t2}, and we present two approaches to resolve this next.
\end{remark}

Remark \ref{rmk:infeasible} implies that the intersection infrastructure, rather than the CAV, is responsible for the feasibility of Problem \ref{prb:reactive}, as depicted in Fig. \ref{fig:state-machine}.
We note that, when Problem \ref{prb:reactive} becomes infeasible, it may take some time for the agent to receive a new crossing time.
In this case, the agent ought to enter a ``safe mode'' by relaxing latest departure constraint.
We consider this for two cases of Definition \ref{def:psi}, for both signalized and unsignalized intersections.

\begin{figure}[ht]
    \centering
    \begin{tikzpicture}[
    block/.style={draw, black,minimum width=3cm, minimum height=1cm},
    line/.style={->, ultra thick}]
        \node[block,fill=green!20] (S) at (0, 0) {Apply $u^*$};
        \node[block,fill=red!20] (P)    at (0,-2) {Get Crossing Time};
        \draw[line] (S.east) --  ++ ( .5, 0) node[below right]{$\underline{u} > \overline{u}$} |-  (P.east);
        \draw[line] (P.west) --  ++ (-.5, 0) node[above left]{$(p_i^k, \underline{t}_i^k, \overline{t}_i^k)$} |- (S.west);
    \end{tikzpicture}
    \caption{A switching system that describes CAV behavior, where the CAV is initialized with tuples of collision node crossing times.
    The CAV generates its control action using Theorem \ref{thm:sln}, and requests an updated schedule if no feasible $u^*$ exists.}
    \label{fig:state-machine}
\end{figure}

\textbf{Unsignalized Intersections} are common in literature that considers $100\%$ penetration of CAVs.
In this case, CAVs request specific crossing times for each conflict node along their trajectory (Definition \ref{def:psi}).
In the case that Problem \ref{prb:reactive} becomes infeasible, there are two possible causes: 1) the crossing time requires an acceleration/deceleration magnitude grater than $u_{\max}$, or 2) the crossing time doesn't sufficiently account for the rear-end safety constraint with a preceding vehicle.
Both cases could occur as a result of poor parameter tuning in the coordinator, and the CAV ought to request a later crossing time that is dynamically feasible and safe.
Under Assumption \ref{smp:comms}, this happens instantaneously.
Otherwise, the CAV ought to enter a safe mode, e.g., emergency stopping or cruising at a constant speed, until a new crossing time is recieved. 

\textbf{Signalized intersections} consist of a single conflict node at the entrance of the intersection.
Furthermore, the feasible crossing time is the time interval where the light is (or will be) green.
In this case, infeasibility of Problem \ref{prb:reactive} implies that the CAV was unable to cross the intersection during the current green light phase, either due to the $u_{\max}$ or rear-end safety constraints being active.
In either case, the CAV can simply request the next green light interval and join the traffic queue to guarantee the feasibility of Problem \ref{prb:reactive}.
For mixed traffic, the signal automatically handles lateral collision avoidance, and the CAV must only consider rear-end safety through \eqref{eq:cbf-rear-end} by measuring $\delta_i(t)$ and $\dot{\delta}_i(t)$.
%

\section{Simulation Results} \label{sec:sim}

To demonstrate the performance of Problem \ref{prb:reactive} at an unsignalized intersection, we compare our approach to the optimal control-based approach in \citep{malikopoulos2021optimal}.
We consider a symmetric intersection with two crossing lanes and a single conflict point at $30$ m.
To generate the trajectories of \citep{malikopoulos2021optimal}, each vehicle generates the optimal trajectory that minimizes
\begin{equation} \label{eq:Ju}
    J_u(u_i, t_i^f) = \frac{1}{2}\int_0^{T} u_i^2\, dt.
\end{equation}
This unconstrained optimal solution is a linear acceleration profile,
\begin{equation}\label{eq:linear-control}
    u_i(t) = a_i t+ b,
\end{equation}
and each CAV $i$ sequentially selects its own arrival time $T$.
Note that $T$ is the the smallest value such that the path satisfies all safety constraints, which makes $\delta_i(t)$ and the safe crossing times explicitly available to all other vehicles.
The resulting trajectories are presented in Fig. \ref{fig:automatica}.

\begin{figure}[ht]
    \centering
    \includegraphics[width=0.95\linewidth]{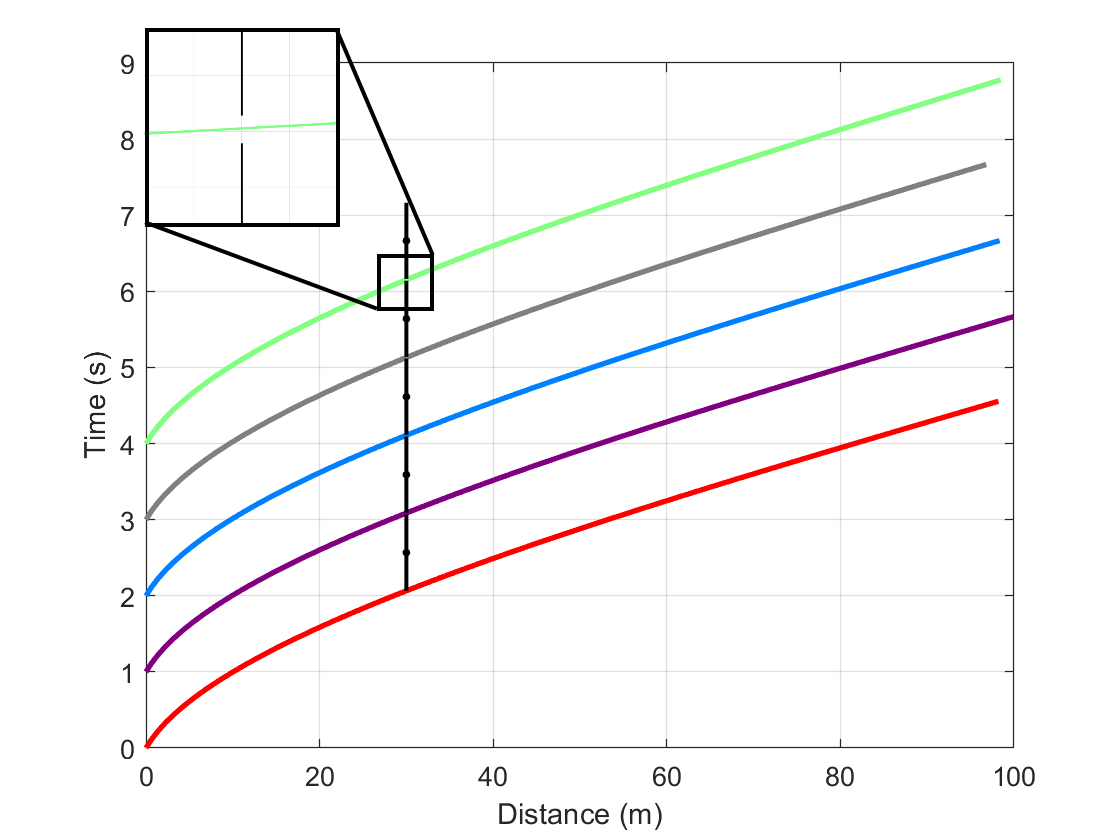}
    \caption{Trajectories for the optimal control-based solution of \citep{malikopoulos2021optimal}. Black vertical lines show the lateral crossing constraints.}
    \label{fig:automatica}
\end{figure}

To create a fair comparison with Problem \ref{prb:reactive}, we manually tuned the simulation parameters based on the optimal control trajectories (see Fig. \ref{fig:automatica}).
Namely, we selected the same initial conditions crossing times, and overall trajectory length for each CAV.
However, this only facilitates an approximate comparison, as \citep{malikopoulos2021optimal} uses a speed dependent time headway to ensure safety, whereas we capture speed dependence through a fixed standstill distance \eqref{eq:cbf-rear-end}.
The resulting parameters for our algorithm using Problem \ref{prb:reactive} are presented in Table \ref{tab:sim_parameters}.

\begin{table}[ht]
    \centering
    \begin{tabular}{cccccc}
        $\alpha$ & $v_d$ & $\gamma$ & $\kappa_T$ & $\kappa_R$ & $u_{\max}$ \\ \toprule
        0.25 s$^{-1}$ & 30 m/s & 1 m & 0.5 & 100 &  25 m/s$^2$
    \end{tabular}
    \caption{Parameters used for the CAV controller in Problem \ref{prb:reactive}.}
    \label{tab:sim_parameters}
\end{table}

As suggested by Remark \ref{rmk:infeasible}, the performance of our proposed approach is closely coupled to the system parameters--including the initial distance to the conflict point, the crossing times, the constraint parameters, and the barrier function coefficients.
In general, the scheduled crossing times directly determine the magnitude of congestion reduction at the intersection.
For individual CAVs, the CBF coefficients in Problem \ref{prb:reactive} determine how aggressively the CAVs will accelerate and decelerate to meet the scheduled crossing times.
Namely, large coefficients will lead to sudden accelerations near the intersection, but the CAV will not be influenced by distant constraints.
In contrast, a small coefficient will cause CAVs to react earlier to potential constraint activations, resulting in a lower overall acceleration.
Matching these parameters to the optimal control approach allows for an approximate comparison between the two methods.
The trajectories generated by Problem \ref{prb:reactive} are shown in Fig. \ref{fig:proposed}.

\begin{figure}[ht]
    \centering
    \includegraphics[width=0.95\linewidth]{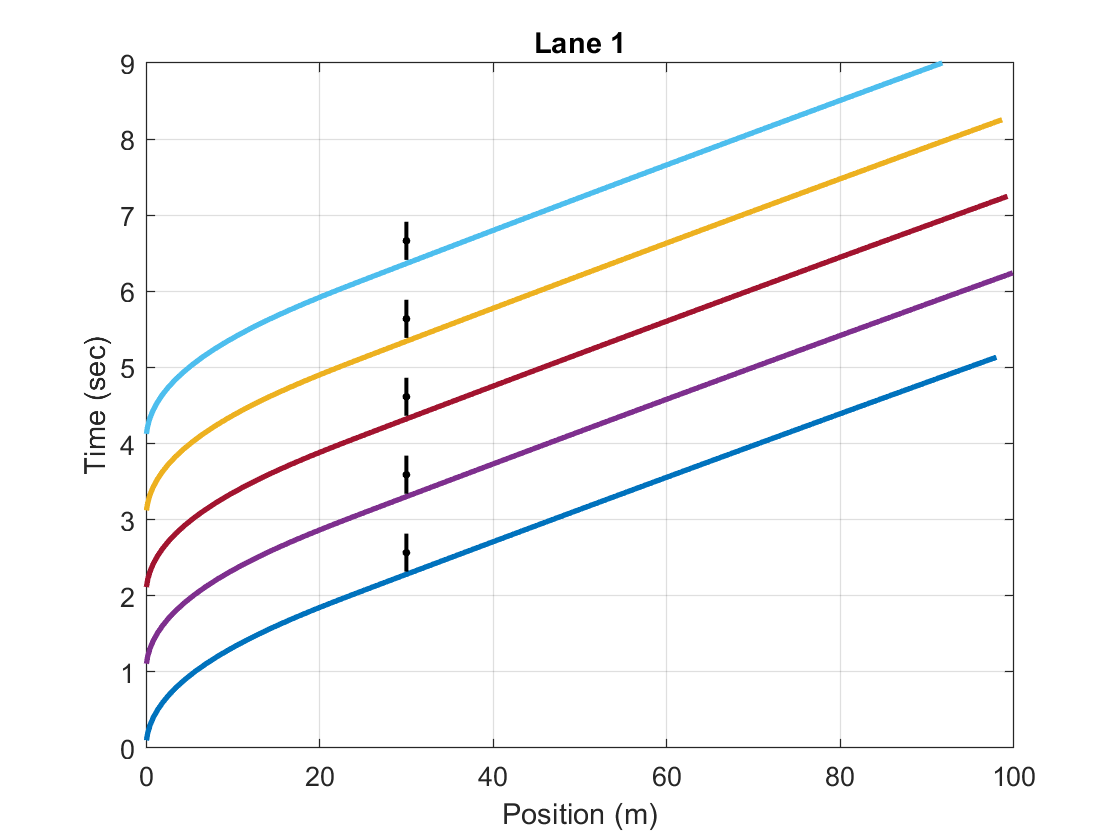}
    \caption{CAV trajectories generated with our proposed approach. Black vertical lines show the lateral crossing constraints.}
    \label{fig:proposed}
\end{figure}

The quantitative results of our analysis are presented in Table \ref{tab:metrics} for both approaches.
We simulated both algorithms for two cases, one with a relatively high acceleration penalty ($\alpha = 0.25$), and one with a lower acceleration penalty ($\alpha = 1.5$).
The column $J_u$ corresponds to the average cost of all vehicles computed with \eqref{eq:Ju}, while $J_{\alpha}$ is computed with \eqref{eq:Jcost}.

\begin{table}[ht!]
    \centering
    \begin{tabular}{r|ccc}
     Algorithm          & $J_u$ & $J_\alpha$ & Time  \\\toprule
     OCP ($\alpha=1.50$) & \textbf{57.2} & 979.9 & 28.7 ms \\
     OCP ($\alpha=0.25$) & \textbf{57.2} & 2056.5 & 28.7 ms \\
     Proposed ($\alpha=1.50$) & 185.9 & \textbf{525.5} & \textbf{1.8 $\mu$s} \\ 
     Proposed ($\alpha=0.25$) & 108.5 & 4030.8 & \textbf{1.8 $\mu$s}
    \end{tabular}
    \caption{Average energy costs and computational time for the optimal control approach \citep{malikopoulos2021optimal} and Problem \ref{prb:reactive}.}
    \label{tab:metrics}
\end{table}

Unsurprisingly, the optimal control approach finds the global minimum for \eqref{eq:Ju}.
Note that the global minimizer of \eqref{eq:Ju} is a linear acceleration \eqref{eq:linear-control} whereas the trajectory generated by Problem \ref{prb:reactive} is an exponential \eqref{eq:mp}.
Thus, we always expect our approach to use more overall control effort for the same constraints and boundary conditions.
Interestingly, when the penalty for applied control effort is high ($\alpha = 0.25$), the optimal control approach also out-performs our proposed approach.
This is because the arrival time constraints of Problem \ref{prb:reactive} perturb the CAVs away from the initial unconstrained trajectory, which would be a global minimizer.
This is clear from Fig. \ref{fig:proposed}, as the CAV trajectories cross the conflict point (located at $30$ m) as late as possible--i.e., the constraint \eqref{eq:cbf-t2} is active along the trajectory.
Thus, we only expect significant fuel economy benefits for long-duration tasks where the impact of engine transients are less significant than tracking an efficient cruising speed.

The quantitative benefit of our approach is the computational time required to generate control actions.
By Theorem \ref{thm:sln}, the control bounds are trivial to compute, and each CAV can generate optimal actions on the order of microseconds.
In contrast, \cite{malikopoulos2021optimal} takes approximately $30$ ms ($5$ orders of magnitude higher) for 10 vehicles at a single crossing node.
For complex node geometries, this can easily increase to the order of seconds, as newly arriving vehicles must plan their trajectories around existing vehicles.

Finally, our proposed approach has the benefits of flexibility and guaranteed feasibility.
One critical drawback of \citep{malikopoulos2021optimal} is that the vehicles \textbf{must} follow a linear acceleration profile, 
which means the vehicles cannot come to a stop (e.g., at a traffic light) and a feasible trajectory may not exist.
Furthermore, planning the entire optimal trajectory means the CAVs need to predict the future trajectories of all other vehicles in the intersection.
While this is possible for 100\% penetration of CAVs, it makes interactions with human driven vehicles a significant open challenge.
In contrast, Problem \ref{prb:reactive} always have a feasible solution (come to a complete stop and wait), and it only requires the current state information of other vehicles. 

\section{Conclusion} \label{sec:conclusion}

In this article, we have developed a long-duration autonomy approach for the operation of CAVs in a transportation network.
We start with an optimal control problem, and follow a principled approach to derive a reactive control law that minimizes energy consumption while guaranteeing safety using CBFs.
We lifted the higher order CBFs to first order using motion primitives. 
Thus, the safety constraints are always compatible with our control effort bounds and the control policy always has a feasible solution.
We also demonstrated the performance of our algorithm in simulation against a state-of-the-art approach based on optimal control.

Future work includes simulation and hardware experiments to demonstrate the capabilities of our proposed controller in many environments--highway, urban, and transitions between them--and mixed traffic scenarios (e.g., see \cite{tzortzoglou2025mixed}).
Incorporating robustness to delays, noise, and unmodeled dynamics is another interesting research direction, as well as considering 2D lane-free environments.
Enabling agents to determine their crossing times in a distributed manner while guaranteeing lateral safety is another intriguing research direction, which aligns with some open problems in multi-agent reinforcement learning.

\bibliographystyle{elsarticle-harv}
\bibliography{my_pubs,mendeley,refs}

\end{document}